\documentclass[AMA,STIX1COL]{WileyNJD-v2}

\usepackage{mathrsfs}
\usepackage{amssymb}
\usepackage{amsfonts}
\usepackage{graphics,graphicx}
\usepackage{amsmath}

\usepackage{algorithm, algorithmicx, algpseudocode}

\usepackage{subfigure}
\usepackage[justification=centering]{caption}

\usepackage{epstopdf}

\newcommand{\A}{\mathscr{A}}
\newcommand{\T}{\top}
\newcommand{\Y}{\mathcal{Y}}

\def\p(#1|#2){p(#1\,|\,#2)}
\def\D(#1\|#2){D(#1 \  \| \ #2)}

\articletype{RESEARCH ARTICLE}%

\begin{document}

\title{SSUE: Simultaneous  State and Uncertainty Estimation for Dynamical Systems\protect\thanks{This work was partially sponsored by the U.S. Army Research Laboratory and was accomplished under Cooperative Agreement Number W911NF-17-2-0138.}}

\author[1]{Hang Geng}

\author[2]{Mulugeta A. Haile}

\author[1]{Huazhen Fang*}

\authormark{Geng, Haile and Fang}

\address[1]{\orgdiv{ Department of Mechanical Engineering}, \orgname{University of Kansas}, \orgaddress{\state{Lawrence KS 66045}, \country{USA}}}

\address[2]{\orgdiv{Vehicle Technology Directorate}, \orgname{US Army Research
Laboratory}, \orgaddress{\state{Aberdeen MD 21005}, \country{USA}}}

\corres{*Huazhen Fang, Department of Mechanical Engineering,
University of Kansas, Lawrence KS 66045, USA. \email{fang@ku.edu}}


\abstract[Summary]{Parameters of the mathematical model describing many practical dynamical systems are prone to vary due to aging or renewal, wear and tear, as well as changes in environmental or service conditions. These variabilities will adversely affect the accuracy of state estimation. In this paper, we introduce SSUE: Simultaneous State and Uncertainty   Estimation for quantifying parameter uncertainty while simultaneously estimating the   internal state of a system. Our approach involves the development of a Bayesian framework that recursively updates the posterior joint density of the unknown state vector and parameter uncertainty. To execute the framework for practical implementation, we develop a computational algorithm based on maximum a posteriori estimation and the numerical Newton's method. Observability analysis is conducted for linear systems, and its relation with the consistency of the estimation of the uncertainty's location is unveiled. Additional simulation results are provided to demonstrate the effectiveness of the proposed SSUE approach.}

\keywords{state estimation, uncertainty estimation,  nonlinear filtering, Bayesian framework, observability analysis}


\maketitle


\section{Introduction}

Many practical engineering systems require state observers that provide an estimate of the internal state of the system using a series of measurements of the input and output. For example, many control systems require state feedbacks to stabilize the system after disturbance. A state observer requires not only a series of sensor data, but also mathematical models describing the physical dynamical system as well as measurement models relating observed measurements to system state. Since the entire process is an approximation of quantities that typically are not directly observable, it entails uncertainties. In dynamic state estimation, uncertainties arise either from disturbance signals or as a result of the dynamic model perturbation. Disturbance signals such as measurement noise, faulty sensor, and random input-output fluctuations are well-known to affect the accuracy of state estimation. Uncertainties due to dynamic model perturbation are primarily due to the discrepancy between the mathematical model and the actual dynamics of the system.  Typical sources of the discrepancy include unmodelled (usually high-frequency) dynamics, neglected model nonlinearities, effects of model over-simplifications (such as deliberate reduced-order models), system-parameter variations (due to aging, wear and tear, and environmental changes), and malicious cyber attacks~\cite{Petersen:1999:RKF,Kim:Springer:2017,Chen:TSMC:2019,Ding:TII:2019}. Reliable estimation of state requires a thorough understanding of these sources of uncertainties as well as ways and means to quantify effects on the accuracy of the state observer. In this paper we attempt to address uncertainties arising from model perturbation; specifically we consider a nonlinear dynamical system whose model parameters are prone to perturbation due to changes in operational environment. Our goal is twofold: first we would like to identify and locate parameters of the mathematical model that are being perturbed, and second we would like to quantify the parameter uncertainties while simultaneously estimating the physical internal state of the dynamical system. Our work is motivated by the need to improve the accuracy of state estimation in practical dynamical systems whose model parameters are prone to changes due to aging or renewal as well as service conditions. Our end result is a robust framework for Simultaneous State and Uncertainty Estimation (SSUE). SSUE is a recursive Bayesian estimation framework with explicit routine for handling parameter uncertainty while simultaneously providing the best estimates of state vector.

\emph{Literature Review}: 
State estimation has been the subject of extensive study for over half a century. Most of the studies are aimed at improving the robustness of the estimation process so as to ensure a higher level of confidence in the estimation accuracy. The crowning achievement of research in this area is the Kalman filter (KF). Indeed,  if a   linear state estimation problem admits the assumption of Gaussian process and noise, then there is no better algorithm that can outperform KF.  However, the model uncertainty  will degrade  the performance of the KF in a variety of real world scenarios, and hence,  other techniques must be employed. Various studies~\cite{Xie:TAC:1994,Xie:SCL:1994,Yang:TAC:2002,Fu:TSP:2001,Petersen:1999:RKF,Bernstein:SCL:1989} present  modified KF method that achieve bounded-error state estimation for bounded model uncertainty. Others~\cite{Nagpal:TAC:1991,Geromel:SIAM:2000,Gao:TAC:2008, Li:AUTO:2012,Chen:TCSI:2002,Khargonekar:IJRNC:1996} present approaches for dealing with parameter uncertainty using robust $\mathcal{H}_\infty$ filtering and mixed $\mathcal{H}_2/\mathcal{H}_\infty$ filtering, which generally consider uncertainty bounded by a convex domain and then performs state estimation through numerical optimization such as convex programming. A somewhat different technique in~\cite{Durola:CDC:2008,Calafiore:AUTO:2004,Ghaoui:TAC:2001,Geromel:TSP:1999,Moheimani:TCSI:1998} implements robust state estimation based on set membership estimation and constrained optimization.

Most of the studies surveyed above often use worst-case model scenarios by using the upper bounds of the uncertainties; this is in addition to an already crude approximation  of uncertainty bounds resulting from an often poor prior knowledge. The end result is an excessively conservative state estimate that is practically unusable. Another approach of dealing with parameter uncertainty is to perform joint state and parameter estimation, which, in a departure from the above studies, seeks to estimate model parameters together with the state. State augmentation  represents a convenient and useful way to accomplish this task. Specifically, it augments the state of a system to include the unknown parameters and then applies a nonlinear KF to the resulting higher-dimensional nonlinear model for state estimation 
~\cite{Ching:PEM:2006,Chowdhary:AST:2010,Chen:JPC:2005}. Dual KF is an alternative to state augmentation, which employs two KFs to estimate the unknown state and parameters alternately, with the state and parameter estimation procedures going forward in parallel and mutually updating each other~\cite{Wan:Wiley:2001b,Wan:NIPS:2000,Vandyke:AAS:2004,Moradkhani:AWR:2005}. It is noted that the effect of parameter uncertainty, together with that of some other kinds of uncertainty (e.g., external disturbances or noises), can be modeled as an unknown input applied to a dynamical system. Then, a problem of much relevance is how to jointly estimate the state and unknown inputs, which has attracted a growing body of work. The study in~\cite{Mendel:TAC:1977} takes a lead with the development of a KF-based approach to estimate the state and external white process noise of a linear system. Recently, many methods have been proposed to handle arbitrary inputs through modifying some existing state estimation techniques. Among them are those based on the 
KF~\cite{Hsieh:TAC:2000,Hsieh:AJC:2010}, moving horizon estimation~\cite{Pina:AUTO:2006},
$\mathcal{H}_\infty$-filtering~\cite{You:AUTOSinica:2008}, sliding mode
observers~\cite{Floquet:IJACSP:2007,Bejarano:IJRNC:2007},  
minimum-variance unbiased estimation~\cite{ Gillijns:AUTO:2007:a, Fang:IJACSP:2011, Fang:AUTO:2012,Yong:AUTO:2016,Shi:AUTO:2016} and Bayesian filtering~\cite{Fang:AUTO:2013,Fang:ACSP:2015,Fang:CEP:2017}. Here, it must be pointed out that the uncertainty considered in these works takes a relatively restrictive form, since it appears in a definite manner as parameters or inputs in a system model. This leaves an open question on how to approach the scenario where the uncertainty's exact presence is unclear.

\emph{Statement of Contributions}: In this paper, we consider the problem of SSUE,  a challenge that has been rarely investigated in the literature and aims to determine not only the state of a system but also parameter uncertainty in terms of magnitude and location in a multi-parameter dynamical model. We focus on a specific SSUE problem, which features a coupling between the state and parameter uncertainty and assumes the uncertainty to belong to some candidate locations. Our study leads to the following contributions. First, it extends the powerful Bayesian approach to develop a Bayesian estimation framework to solve the SSUE problem. This framework includes two parts. The first part deals with estimation of state and uncertainty conditioned on the uncertainty`s location. Here we leverage the notion of Bayesian multi-model estimation to locate uncertainty. Our second contribution lies in the development of a computational SSUE algorithm, which integrates maximum a posteriori (MAP) estimation and numerical optimization to implement the proposed estimation framework. Finally, we conduct  theoretical analysis of the observability and consistency of the estimation of the uncertainty's location under the framework. We develop a definition of joint observability for SSUE and prove that the multi-model-based Bayesian estimation of the uncertainty's location is guaranteed to be consistent under certain conditions.

\emph{Organization}: The remainder of this paper is organized as follows. Section~\ref{Problem-Formulation} presents the formulation of the SSUE problem investigated in this work and the Bayesian estimation framework.  Section~\ref{Observability-Analysis}  defines observability for the SSUE problem and analyzes the consistency of the estimation of the uncertainty's location.   Section~\ref{SSUE-Algorithm} presents a computational implementation of the Bayesian framework by combining the MAP estimation and Newton's method. A simulation example is provided in Section~\ref{Simulation-Example} to evaluate the proposed SSUE algorithm. Finally, Section~\ref{Conclusion} provides concluding remarks. Most of the proofs are gathered in the Appendix.

\emph{Notation}: The notation used throughout the paper is fairly standard.
The superscripts `$\top$' and `$-1$' stand for matrix transposition and
inverse, respectively; $p$ denotes  the probability density function (pdf),
probability mass function (pmf), or mixed pdf-pmf for random variables or vectors, depending on the context; $%
\mathbb{R}$$^{n}$ denotes the $n$-dimensional Euclidean space; $\mathcal{N}%
\left( \bar{x},\Sigma \right) $ is a Gaussian function with mean $\bar{x}$
and covariance $\Sigma $; $\nabla _{x}y$ represents the derivative of $y$
with respect to $x$. Finally, ``a.s.'' denotes almost sure convergence.

\section{The SSUE Problem and Bayesian Framework}\label{Problem-Formulation}

Consider the following dynamical system subject to parameter uncertainty:
\begin{align}  
\left\{ \begin{aligned} x_{k+1} &=(A+\delta\mathscr{A})x_{k}+w _{k}, \\
y_{k} &= h( x_{k}) +v _{k}, \end{aligned}\right.  \label{mas1}
\end{align}%
where $x_{k}\in \mathbb{R}^{n}$ is the state of the system, $y_{k}\in \mathbb{R}%
^{p}$ is  the measurement,   $w
_{k}\in \mathbb{R}^{n}$ and $v _{k}\in \mathbb{R}^{p}$   are the process and measurement noises, respectively, and the nonlinear mapping $h: \mathbb{R}^n \rightarrow \mathbb{R}^p$ represents the measurement function. The initial state $x_0$ and the noises $w_k$ and $v_k$ are white Gaussian random vectors with zero mean  and covariances of $P_0>0$, $Q_k\geq 0$ and $R_k >0$, respectively, and are independent of each other.
 In addition,  $\delta  \in \Delta  $ represents  the parameter uncertainty, where $\Delta \subseteq \mathbb{R}$ can be  an interval or a union of multiple intervals in the   real axis.  The matrix $\mathscr{A}\in \mathbb{R}^{n\times n}$ associates $\delta$ with the affected parameters, which can be regarded as  a ``location matrix''.   
 For  simplicity and without loss of generality, we consider that the elements of $\A$ are either 0 or 1, with 1 imposing $\delta$ on the corresponding parameter in $A$. Even though $\A$ is unknown, it is reasonable to assume that some candidate locations of the uncertain parameters can be known \textit{a priori}. Hence, we can let $\A$ belong to a set of  location matrices, i.e., $\A \in \mathbb{A}= \{\A_i, i=1,2,\ldots,M\}$, where $M$ is the degree of the set. We denote $\mathbb{M} = \{1,2,\ldots,M\}$. An example of $\mathbb{A}$ for a two-dimensional system is $\mathbb{A} =   \{\A_1,\A_2,\A_3\}$, where
\begin{align*}
 \A_1=
\begin{bmatrix}
0 & 1 \\
0 & 0%
\end{bmatrix}%
, \ %
\A_2 = \begin{bmatrix}
0 & 0 \\
1 & 0%
\end{bmatrix}%
, \ %
\A_3 = \begin{bmatrix}
1 & 0 \\
0 & 1%
\end{bmatrix}.%
\end{align*}
 Note that another set $\mathbb{A}$ can be chosen if the \textit{a priori}  knowledge suggests other possible locations of the uncertainty.    The above system description can be used to capture parameter drift, model mismatch, and even disturbances on dynamics  due to cyber attacks to a cyber-physical system. 
Based on this, the objective of SSUE  is to simultaneously estimate $x_k$ and $\delta$ and determine $\A$ using the model in~\eqref{mas1} and the output  data stream $\mathcal{Y}_k=\{y_0, y_1, \ldots, y_k\}$.

\begin{remark}
The above describes  a basic model for the study of SSUE, with the process equation linear for zero $\delta$ and the measurement equation nonlinear. We consider this model for two reasons. First, it conveniently and explicitly shows the presence of parameter uncertainty in the system, in addition to its direct effect on the process evolution. This will help us expound the estimation design and analysis   in a clear-cut manner. Second, the process equation is still nonlinear due  the existence of the unknown $\delta$, making it not that different from a general nonlinear   equation, especially since SSUE design is our focus. 
The model formulation  also  readily allows for extension in different ways. First, we can generalize $\delta$ from a scalar to a vector and accordingly define a set of  location matrices for each element of a vector-based $\delta$. 
Second,   we can consider a nonlinear process equation  and the scenario when both the process and measurement equations are affected by uncertainty. In this case, we can still define a location set by making the   potential locations of the uncertainty in the nonlinear functions $f$ and $h$. When such a more sophisticated model is used, we can accommodate it by upgrading most of the developments that follow in this paper, particularly the designs of the Bayesian estimation  framework and computational algorithm.  
\end{remark}

We intend to lay out a Bayesian estimation framework for the SSUE problem described above. The Bayesian approach has been used to address  a wide range of estimation problems~\cite{Fang:JAS:2018}. The idea at its core is to sequentially update the probabilistic belief about an unknown variable given the measurement data stream. Here, we   build on this thinking to  develop the Bayesian  SSUE in two integral parts:
\begin{enumerate}

\item[1)] The first part is about inferring $x_k$ and $\delta$ given a possible location of the uncertainty, $\A_i$, and the measurement data $\Y_k$. This implies a need to update the   \textit{ a posteriori} pdf $\p(\delta,x_k|\A_i, \Y_k)$ for $i \in \mathbb{M}$. 

\item[2)] The second part  is concerned with determining $\A$ from among the candidates.  We   can consider updating the  pmf  $\p(\A=\A_i|\Y_k)$ (denoted as  $\p(\A_i|\Y_k)$ in sequel for notational conciseness) for $i \in \mathbb{M}$, with the hope that the probability associated with the true location of the uncertainty will surpass that of the other locations.  

\end{enumerate}
Based on 1) and 2), one can  further  exploit them to update $\p(\delta,x_k | \Y_k)$.  This procedure, as is observed, pertains to   multi-model   estimation, which has been significantly studied in the literature~\cite{Li-Jilkov:2005,Xu:TAES:2016}. Nonetheless, a conventional multi-model state estimation scheme is not applicable here, because of the existence of the state-coupled uncertainty. We, therefore, custom-design a  Bayesian estimation framework for SSUE, as  shown  in the following theorem.
\begin{theorem}\label{Bayesian-Paradigm}
For the system in~\eqref{mas1}, the  Bayesian paradigm for SSUE  is given by
\begin{subequations}\label{Bayesian-Update}
\begin{align}
 \p(\delta ,x_{k}|\mathscr{A}_{i},\mathcal{Y}_{k-1}) 
&=\int \p( x_{k}  | \delta ,x_{k-1}, \A_i ) \p( \delta ,x_{k-1} | 
\A_i,\mathcal{Y}_{k-1} ) dx_{k-1}, \label{mas6}\\
\p( \delta ,x_{k}  |   \mathscr{A}_{i},\mathcal{Y}_{k} )
& \varpropto \p( y_{k}  |  x_{k}  ) \p( \delta
,x_{k} \, | \, \mathscr{A}_{i},\mathcal{Y}_{k-1} ) ,  \label{mas7}\\
 \p(\mathscr{A}_{i}|\mathcal{Y}_{k})& =\frac{ \p(y_{k}|\mathscr{A}_{i},\mathcal{Y}_{k-1}) \p(\mathscr{A}_{i}|\mathcal{Y}_{k-1})}{\sum_{i=1}^{M} \p(y_{k}|\mathscr{A}_{i},\mathcal{Y}_{k-1}) \p(\mathscr{A}_{i}|\mathcal{Y}_{k-1})},
\label{mas5}\\
\p(\delta ,x_{k}|\mathcal{Y}_{k})& =\sum_{i=1}^{M} \p(\delta ,x_{k}|\mathscr{A}_{i},\mathcal{Y}_{k}) \p(\mathscr{A}_{i}|\mathcal{Y}_{k}), \label{mas4} 
\end{align}
\end{subequations}
for $i\in \mathbb{M}$.
\end{theorem}
\begin{proof}
See Appendix.\ref{Bayesian-Paradigm-Proof}.
\end{proof}

Theorem~\ref{Bayesian-Paradigm} characterizes the recursive Bayesian update   for the  pdf/pmf  of $x_k$, $\delta$, and $\A_i$ conditioned on the measurement data, leading to  a conceptual framework for SSUE.   In particular,~\eqref{mas6} shows the prediction of the joint conditional pdf of $\delta$ and $x_k$,~\eqref{mas7} performs the update of the pdf when the new   data $y_k$ is available, and~\eqref{mas5}-\eqref{mas4} evaluates the possibility of each   $\A_i$ and enables a fusion of the  pdf-pmf's based on different $\A_i$'s. It should be noted that,
because of the nonlinear nature of the model in~\eqref{mas1}, it is usually impossible to derive a closed-form solution as an estimator to the  Bayesian framework proposed in Theorem~\ref{Bayesian-Paradigm}.  We hence will  resort to numerical optimization  to computationally execute the framework.  Section~\ref{SSUE-Algorithm} will present an algorithm based on the  Newton's iterative method to achieve this end.

\section{Observability and Consistency Analysis}\label{Observability-Analysis}

In this section, we investigate  the   observability for SSUE and the consistency of the estimation of $\A$. To formulate a tractable analysis, we restrict our attention to a reduced version of the system in~\eqref{mas1}:
\begin{align}   \label{Linear-model}
\left\{ \begin{aligned} x_{k+1} &=(A+\delta\mathscr{A})x_{k}+w _{k}, \\
y_{k} &= C x_{k}+v _{k},
\end{aligned}\right. 
\end{align}%
where the measurement equation is linear.

Now, let us define the joint observability for $x_k$, $\delta$, and $\A$ and derive corresponding   testing conditions. 
With inspirations drawn from~\cite{Vidal:CDC:2002}, we propose the following definitions:

\begin{definition}
\label{def1} (Indistinguishability).  The triplet pairs $(x_0,\delta,
\mathscr{A})$ and $(x_0^{\prime},\delta^{\prime},\mathscr{A}^{\prime})$ for $\delta,\delta^\prime \in \Delta $ and $\A,\A^\prime \in \mathbb{A}$
are said to be indistinguishable on the interval $[0,K]$ for $K \in \mathbb{Z}^+$, if the corresponding output
sequences $\{y_0, y_1,\ldots,y_K\}$ and $\{y_0^{\prime}, y_1^{\prime}, \ldots, y_K^{\prime}\}$ when $w_k=0$ and $v_k=0$  are equal. We denote the set of triplet pairs that are indistinguishable from $(x_0,\delta,
\mathscr{A})$   as $\mathcal{I}(x_0,\delta,\mathscr{A})$.
\end{definition}

\begin{definition}(Observability). The triplet pair  $(x_0,\delta,\mathscr{A})$
is observable on the interval $[0, K]$ if $\mathcal{I}(x_0,\delta,%
\mathscr{A})=(x_0,\delta,\mathscr{A})$. When any admissible pair is
observable,  we   say that the system in~\eqref{Linear-model} is observable.\label{Obs-Def}
\end{definition}

Note that the definitions above are made in terms of the initial state $x_0$ and require the output sequences to differ when given a different $(x_0,\delta,\A)$. The following theorem provides an observability testing condition accordingly.

\begin{theorem}\label{Observability-testing}
\label{lem1} The system in~\eqref{Linear-model} is observable if and  only if there exists $N\in \mathbb{Z}^+$ and $N \leq K$ such that for $k\geq N$,
\begin{equation}  \label{rank-testing}
\text{rank}\left(
\begin{bmatrix}
\mathcal{O}_{k} (\delta,\mathscr{A} ) & \mathcal{O}_{k}(\delta^{\prime},%
\mathscr{A}^{\prime})%
\end{bmatrix}%
\right)=2n,
\end{equation}
where $\delta, \delta' \in \Delta$, $\A, \A'\in\mathbb{A}$, $(\delta, \A) \neq (\delta', \A')$, and
\begin{align*}
\mathcal{O}_{k}(\delta,\mathscr{A})=
\begin{bmatrix}
C^{\T} & \left(C(A+\delta \mathscr{A})\right)^{\T} & \cdots & \left(C(A+\delta %
\mathscr{A})^{k}\right)^{\T}%
\end{bmatrix}%
^{\T}.
\end{align*}
Furthermore, the true $x_0^*$, $\delta^*$ and $\A^*$ can be reconstructed by
\begin{subequations}
\begin{align}
\delta^*, \A^* &= \left\{\delta, \A: \mathrm{rank}\left(\left[\begin{matrix}\mathcal{O}_{k}(\delta, \A) & \mathcal{Y}_k^*\end{matrix}\right]\right) = n, \delta \in \Delta, \A \in \mathbb{A} \right\}, \\  x_0^* &= \left[\mathcal{O}_{k}^\top(\delta^*,\mathscr{A}^*) \mathcal{O}_{k}(\delta^*,\mathscr{A}^*) \right]^{-1} \mathcal{O}_{k}^\top(\delta^*,\mathscr{A}^*) \mathcal{Y}_k^*,
\end{align}
\end{subequations}
where $\mathcal{Y}_{k}^*$ is the output   sequence generated based on $x_0^*$, $\delta^*$ and $\A^*$ when $w_k=0$ and $v_k=0$.
\end{theorem}

\begin{proof}
It is obvious  that $\mathcal{Y}_k$ lies in the column space of  $\mathcal{O}_k$ since  $\mathcal{Y}_k = \mathcal{O}_k(\delta, \A) x_0$. Hence, the necessary and sufficient condition for observability is that      the column spaces of different $\mathcal{O}_k$'s due to different $\delta$ and $\A$ must not overlap. With this notion, the result above is straightforward, so a detailed proof is    omitted here. 
\end{proof}

\begin{remark}
According to the above, the output sequence  is unique for a given triplet pair $(x_0,\delta,\A)$ for the system in~\eqref{Linear-model} under noise-free evolution. This thus provides an ability to uniquely determine $x_0$, $\delta$, and $\A$, which then allows one to reconstruct $x_k$ through time. The observability  definition can be considered as an extension of the standard observability for a linear deterministic system~\cite{Chen:OUP:2012}. An interesting question that remains is the development of weaker properties weaker, e.g., detectability, which will require  future work.  
\end{remark}

Based on the observability analysis, we   go forward to investigate  the consistency of the estimation of $\A$, which denotes the uncertainty's location. 
The overall aim of the consistency analysis here is to determine whether the true location, denoted as $\A_t$ for $t \in \mathbb{M}$, will stand out with the highest probability if we  evaluate $\p(\A_i | \Y_k)$  for each $i \in \mathbb{M}$. 
To begin with, we make the following assumption:
\begin{assumption}\label{Location-Equal-Porbability}
The \textit{a priori} probability  $p(\A_i) = \frac{1}{M}$ for each $i\in \mathbb{M}$.
\end{assumption}
Assumption~\ref{Location-Equal-Porbability} sets an equal initial probability for every possible location. With this, one can easily find out that the comparison between $\p(\A_t | \Y_k)$ and  $\p(\A_i | \Y_k)$ for $i  \neq t$ is equivalent to the one between $\p( \Y_k | \A_t) / \p( \Y_k | \A_i)$, by   Bayes' rule. As such, we consider the Kullback-Leibler (KL) divergence between   $\p(\mathcal{Y}_{k}|\mathscr{A}_{t})$ and $\p(\mathcal{Y}_{k}|%
\mathscr{A}_{j})$:
\begin{align*}    
\D(\mathscr{A}_{t} \| \mathscr{A}_{j}) =\mathbb{E}_{t} \left(
\log h_j^t  (\Y_k  )
\right) = \int \p(\mathcal{Y}_{k}|\mathscr{A}%
_{t})\log \frac{ \p(\mathcal{Y}_{k}|\mathscr{A}_{t})}{ \p(\mathcal{Y}_{k}|%
\mathscr{A}_{j})}d {\mathcal{Y}_{k}},
\end{align*} 
and define
\begin{align*}
h_{j}^{i} (\mathcal{Y}_{k} )= \frac{ \p( \mathcal{Y}_{k}|\mathscr{A}_{i} ) }{\p( \mathcal{Y}_{k}|\mathscr{A}_{j} ) }, \  \ h_{j}^{i} (y_{k} \, | \, \mathcal{Y}_{k-1} )= \frac{ \p( y_{k}|\mathscr{A}_{i},\mathcal{Y}_{k-1}) }{\p( y_{k}|\mathscr{A}_{j},\mathcal{Y}_{k-1} ) },
\end{align*}
Proceeding further, we   put together two more assumptions:

\begin{assumption}
\label{assum1} The considered system in~\eqref{Linear-model} is observable. 
\end{assumption}

\begin{assumption}
\label{assum2} The stochastic process $\left\{
\log h_i^t  (y_k  \, | \, \Y_{k-1}  ) \right \}$ is ergodic for each $i\in \mathbb{M}$.
\end{assumption} 
The theorem below summarizes the main result about the consistency for the estimation of $\A$. 

\begin{theorem}\label{Consistency}
Consider the system in~\eqref{Linear-model}. If Assumption~\ref{assum1} holds, then 
\begin{subequations}
\begin{align} \label{KL-Divergence}
\D(\mathscr{A}_{t} \| %
\mathscr{A}_{i})&>0, 
\end{align}
\end{subequations}
for $t \neq i$.
\setcounter{equation}{5}
Further, if Assumptions~\ref{Location-Equal-Porbability} and~\ref{assum2} also hold, then
\begin{subequations}    \refstepcounter{equation}
\begin{align}
\label{Ratio-Larger-than-1}
{\lim_{k\rightarrow \infty}}\frac{\p(\mathscr{A}_{t} | \mathcal{Y}_{k} )}{\p( \mathscr{A}_{i} | \mathcal{Y}_{k})}&>1,\\ \label{Ratio-infty}
\underset{%
k\rightarrow \infty }{\lim }\frac{\p( \mathscr{A}_{t} | \mathcal{Y}_{k})}{\p( \mathscr{A}_{i} | 
\mathcal{Y}_{k})}&=\infty , \ \mathrm{a.s.}
\end{align}
\end{subequations}
\end{theorem}

\begin{proof}
See Appendix.\ref{Consistency-Proof}
\end{proof}

Theorem~\ref{Consistency}  unveils  a three-fold, progressive result. First,~\eqref{KL-Divergence} shows that   $\p(\Y_k | \A_t)$ differs from $\p(\Y_k | \A_i)$ for $i \neq t$ when the system is observable.  This is because the output sequence based on $\A_t$ will have statistics differentiable from that based on any other $\A_i$. Then,~\eqref{Ratio-Larger-than-1}   highlights that $\p(  \A_t | \Y_k ) $ can surpass $\p(  \A_i | \Y_k )$ for $i\neq t$ as $k \rightarrow \infty$, suggesting that    the true location $\A_t$ will be favored in the evaluation of $\p(  \A_i | \Y_k )$ for $i\in\mathbb{M}$. Finally,~\eqref{Ratio-infty} implies that  $\p( \A_t | \Y_k)$ and $\p( \A_i | \Y_k)$ for $i \neq t$ will   approach 1 and 0 as $k\rightarrow \infty$, respectively. Loosely speaking, Theorem~\ref{Consistency} suggests that, by evaluating $\p(\A_i | \Y_k)$ for $i \in \mathbb{M}$, one can  almost surely determine $\A_t$ from the candidate locations for an observable system.
  Note that   Assumptions~\ref{assum1}-\ref{assum2} allow  us to develop the above result, which has no equivalence in the literature to our knowledge and is worthwhile to know. It will be interesting and useful to investigate whether  Assumptions~\ref{assum1}-\ref{assum2} can be relaxed to obtain more general conditions, which will motivate our further research effort. 


\section{Computational Aspects of Bayesian SSUE}\label{SSUE-Algorithm}

This section focuses on developing a computational SSUE algorithm to execute the framework proposed in Section~\ref{Problem-Formulation}.   
Let us begin by  making the following assumption:

\begin{assumption}
\label{assum4}
Each of the pdf's $\p(y_{k}|x_{k}),$ $\p(\delta ,x_{k}|\mathscr{A}%
_{i},\mathcal{Y}_{k})$ and $\p(\delta ,x_{k}|\mathscr{A}_{i},\mathcal{Y}%
_{k-1})$ can be replaced by a Gaussian distribution:
\begin{subequations}
\begin{align}\label{P_y_x_Gaussian}
\p(y_{k}|x_{k})& \sim \mathcal{N}\left( h(x_{k}),R_{k}\right) , \\ \label{P_delta_x_Y_k_Gaussian}
\p(\delta ,x_{k}|\mathscr{A}_{i},\mathcal{Y}_{k})& \sim \mathcal{N}\Bigg(%
\begin{bmatrix}
\hat{\delta}_{i,k|k} \\
\hat{x}_{i,k|k}%
\end{bmatrix}%
,%
\begin{bmatrix}
P_{i,k|k}^{\delta } & P_{i,k|k}^{\delta x} \\
(P_{i,k|k}^{\delta x})^{\T} & P_{i,k|k}^{x}%
\end{bmatrix}%
\Bigg), \\ \label{P_delta_x_Y_k-1_Gaussian}
\p(\delta ,x_{k}|\mathscr{A}_{i},\mathcal{Y}_{k-1})& \sim \mathcal{N}%
\Bigg(%
\begin{bmatrix}
\hat{\delta}_{i,k|k-1} \\
\hat{x}_{i,k|k-1}%
\end{bmatrix}%
,%
\begin{bmatrix}
P_{i,k|k-1}^{\delta } & P_{i,k|k-1}^{\delta x} \\
(P_{i,k|k-1}^{\delta x})^{\T} & P_{i,k|k-1}^{x}%
\end{bmatrix}%
\Bigg),
\end{align}%
\end{subequations}
where $\hat{\delta}_{i,k|l}$ and $\hat{x}_{i,k|l}$ are the 
estimates of  $\delta $ and  $x_{k}$ at time $k$ when given  $\mathscr{A}_{i}$  and $\Y_l$  with the associated covariances denoted as
$P_{i,k|l}^{\delta }$ and $P_{i,k|l}^{x}$, respectively, and $P_{i,k|l}^{\delta x}$  is the corresponding cross-covariance.
\end{assumption}

Note that if complete descriptions  of $\p(\delta ,x_{k}|%
\mathscr{A}_{i},\mathcal{Y}_{k})$ and $\p(\mathscr{A}_{i}|\mathcal{Y}_{k})$
are available, $\hat{\delta}_{i,k|k}$ and $\hat{x}_{i,k|k}$ can be easily
obtained using straightforward methods like conditional  mean calculation. However, the nonlinear coupling of $\delta ,x_{k}$, and $%
\mathscr{A}_{i}$ makes this ideal unrealistic.   Assumption \ref{assum4} is thus presented to mitigate this difficulty, through  approximating   $\p(y_{k}|x_{k})$, $\p(\delta ,x_{k}|\mathscr{A}%
_{i},\mathcal{Y}_{k})$, and $\p(\delta ,x_{k}|\mathscr{A}_{i},\mathcal{Y}%
_{k-1})$ each as a Gaussian distribution with the same mean and covariance.   Gaussianity assumptions akin to this one represent a common treatment for   nonlinear estimation problems, which can be found in diverse literature~\cite{Anderson:1979,Ito-TAC:2000}. 

We now set off to develop a computational algorithm based on~\eqref{mas6}-\eqref{mas4} and Assumption~\ref{assum4}. Our first step is to predict $\delta$ and $x_k$ given $\Y_{k-1}$ and $\A_i$, leveraging~\eqref{mas6}. We consider the MAP estimation  problem below:
\begin{equation}
\begin{bmatrix}
\hat{\delta}_{i,k|k-1} \\
\hat{x}_{i,k|k-1}%
\end{bmatrix}%
=\text{arg}~\underset{\delta ,x_{k}}{\text{max}}~\p(\delta ,x_{k}|\mathscr{A}%
_{i},\mathcal{Y}_{k-1}) , \label{Prediction-MAP}
\end{equation}
The solution is given in the following lemma.
\begin{lemma}\label{Prediction-Lemma} 
Suppose that Assumption~\ref{assum4} holds. An approximate solution to the prediction problem in~\eqref{Prediction-MAP} is 
\begin{subequations}\label{Prediction-Formula}
\begin{align}
\hat{\delta}_{i,k|k-1}& =\hat{\delta}_{i,k-1|k-1},  \label{mas9} \\
\hat{x}_{i,k|k-1}& =(A+\hat{\delta}_{i,k-1|k-1}\mathscr{A}_{i})\hat{x}%
_{i,k-1|k-1},  \label{mas10}
\end{align}%
\end{subequations}
\setcounter{equation}{8}
with the associated covariances being
\begin{subequations} \refstepcounter{equation}
\begin{align}
P_{i,k|k-1}^{\delta }&= P_{i,k-1|k-1}^{\delta },  \label{mas11} \\
P_{i,k|k-1}^{x}&= F_{k-1}
\begin{bmatrix}
P_{i,k-1|k-1}^{\delta } & P_{i,k-1|k-1}^{\delta x } \cr  (P_{i,k-1|k-1}^{\delta x})^\T & P_{i,k-1|k-1}^{ x }
\end{bmatrix} F_{k-1}^\top  +Q_k,\\
P_{i,k|k-1}^{\delta x  }&= 
\begin{bmatrix}
 P_{i,k-1|k-1}^{\delta } & P_{i,k-1|k-1}^{ \delta x } 
\end{bmatrix}F_{k-1} ^\T,
\end{align}
\end{subequations}
where 
\begin{align*}
F_{k-1} =  \begin{bmatrix} \mathscr{A}_{i}\hat{x}%
_{i,k-1|k-1} &  A+\hat{\delta}_{i,k-1|k-1}\mathscr{A}_{i}   \end{bmatrix} .
\end{align*}
\end{lemma}

\begin{proof}
See Appendix.\ref{Prediction-Proof}.
\end{proof}

Based on  the passing from $\p(\delta, x_k | \A_i, \Y_{k-1})$ to $\p(\delta, x_k | \A_i, \Y_{k})$  shown in~\eqref{mas7},   the arrival of $y_k$ can allow us to update the prediction. We consider the problem
\begin{align}\label{Update-Estimation}
\begin{bmatrix}
\hat{\delta}_{i,k|k} \\
\hat{x}_{i,k|k}%
\end{bmatrix}%
=\arg\max_{\delta ,x_{k}}~\p(\delta ,x_{k}|\mathscr{A}%
_{i},\mathcal{Y}_{k}). 
\end{align}%
Under Assumption~\ref{assum4}, one can convert~\eqref{Update-Estimation}   into a nonlinear optimization problem, which can be addressed by the  Newton's method. This leads to a solution provided below, where  we denote $\xi _{k} = 
\begin{bmatrix}
\delta ^{\T} & x_{k}^{\T}%
\end{bmatrix}^\T$ for notational simplicity.  

\begin{lemma}
\label{Update-Lemma} Suppose that Assumption~\ref{assum4} holds. A Newton's iterative update solution to \eqref{Update-Estimation} is given by
\begin{subequations}\label{Update-Formula}
\begin{align} \label{Newton-Iterative-Update}
\hat{\xi}_{i,k|k}^{(\ell+1)}=\hat{\xi}_{i,k|k}^{(\ell)}-\left[ \nabla _{\xi
}^{\T}r_{i} \left(\hat{\xi}_{i,k|k}^{(\ell)}\right) \cdot \nabla _{\xi }r_{i} \left(\hat{\xi}_{i,k|k}^{(\ell)}\right) + S_{i} \left(\hat{\xi}_{i,k|k}^{(\ell)}\right)  \right]^{-1} \nabla _{\xi }^{\T}r_{i} \cdot r_{i}\left( \hat{\xi}%
_{i,k|k}^{(\ell)}\right) ,  
\end{align}%
\end{subequations}\setcounter{equation}{10}
where the superscript $(\ell)$ counts the iteration step, and
\begin{align*}
r_{i}(\delta ,x_{k}) &= 
\begin{bmatrix}
R_{k}^{-{1 \over 2}} (y_{k}- h(x_{k})) \\
\left(P_{i,k|k-1}^\xi \right)^{-{1 \over 2}}(\xi_{k}-\hat{\xi}_{i,k|k-1})%
\end{bmatrix}, \
P_{i,k|k-1}^\xi = \begin{bmatrix}
P_{i,k|k-1}^{\delta } & P_{i,k|k-1}^{\delta x} \\
\left(P_{i,k|k-1}^{\delta x}\right)^{\T} & P_{i,k|k-1}^{x}%
\end{bmatrix},\\ 
\nabla _{\xi }r_{i}&=%
\begin{bmatrix}
0 & - R_k^{- {1 \over 2}} C^{(\ell)} \cr \multicolumn{2}{c}{ \left(P_{i,k|k-1}^\xi \right)^{-{1 \over 2}}}
\end{bmatrix}, \  C^{(\ell)}= \left. \nabla_x h \right|_{\hat{x}_{i,k|k}^{(\ell)}},  \ S_{i} \left(\hat{\xi}_{i,k|k}^{(\ell)}\right) = \sum_{j=1}^{p+n+1} r_{i,j}\left(\hat{\xi}_{i,k|k}^{(\ell)}\right) \nabla^2 r_{i,j}\left(\hat{\xi}_{i,k|k}^{(\ell)}\right).
\end{align*}
Here, $r_{i,j}$ is the $j$-the element of the column vector $r_i$.
The associated estimation error covariance  is 
\begin{subequations} \refstepcounter{equation}
\begin{align}
P_{i,k|k}^\xi 
= 
\left[
H+ (P_{i,k|k-1}^\xi)^{-{1 }}
\right]^{-1}
,  \label{mas16}
\end{align}
\end{subequations}
with
\begin{align*}
H = \begin{bmatrix}
0 & 0 \cr 0 & \left( C^{(\ell_{\max})}\right)^\T R_k^{-1} C^{(\ell_{\max})}
\end{bmatrix},
\end{align*}
where $\ell_{\max}$ is the maximum iteration number.
\end{lemma}

\begin{proof}
See Appendix.\ref{Update-Proof}.
\end{proof}

To initialize the running of~\eqref{Newton-Iterative-Update},  we can let  $\hat{\xi}_{i,k|k}^{(0)}$ be $\hat{\xi}_{i,k|k}^{(0)}=%
\begin{bmatrix}
\hat{\delta}_{i,k|k-1}^\top & \hat{x}_{i,k|k-1}^\top%
\end{bmatrix}^\top%
$ for convenience. The termination of its running depends on a pre-set stop condition, which can be based on a maximum iteration number $\ell_{\max}$ or a threshold for the difference between two consecutive iterations. The   estimate obtained at the final iteration is assigned to be $\hat \xi_{i, k|k}$.

\begin{remark}
The Newton's method is used here because it often demonstrates good convergence properties --- at least quadratic convergence in a local neighborhood of the true solution. In practice, one can slightly perturb $Q_k$ such that $Q_k>0$ to improve the numerical stability.  Meanwhile, it should be noted that  the numerical optimization for~\eqref{Update-Estimation} can be performed using other approaches. For instance, when the problem is mildly   nonlinear, we  can neglect the second-order term $S_i$, which then   yields a Gauss-Newton update.  In addition, the Levenberg-Marquardt   or conjugate gradient methods can be other options for the optimization procedure.   
\end{remark}

Moving forward,  we   seek to evaluate $\p(\A_i | \Y_k)$ for each $ i \in \mathbb{M}$ based on~\eqref{mas5}  and estimate $\A$ through
\begin{align}\label{Location-Identification-Initial}
\hat \A_k  = \arg\max_{\A_i \in \mathbb{A} } \p(\A_i | \Y_k).
\end{align}
Following this, we can perform multi-model fusion of the $\delta$ and $x_k$ estimation conditioned on different $\A_i$'s using~\eqref{mas4}. The next lemma shows the  procedure to accomplish these aims.

\begin{lemma}\label{Fusion-Lemma}
Suppose that Assumption~\ref{assum4} holds. The estimation of $\A$ at time $k$ is 
\begin{subequations}\label{Location-Evaluation}
\begin{align}\label{Location-Identification}
\hat \A_k = \A_\rho, \ \mathrm{with} \  \rho = \arg\max_{i\in \mathbb{M} } \mu_{i,k},
\end{align}
\end{subequations}\setcounter{equation}{12}
where $\mu_{i,k} = \p(\A_i | \Y_k)$ and $\lambda_{i,k} = \p(y_k | \A_i,  \Y_{k-1})$   are expressed as, respectively,
\begin{subequations}\refstepcounter{equation}
\begin{align}\label{mu-update}
\mu _{i,k}&=\frac{\lambda _{i,k}\mu _{i,k-1}}{\overset{M}{\underset{i=1}{\sum
}}\lambda _{i,k}\mu _{i,k-1}},\\ \label{lambda-update}
\lambda _{i,k}=& {(2\pi)^{- {p\over 2} } |\Gamma _{i,k}|^{- {1 \over 2}}}\exp \left[-\frac{1}{2}\left(y_{k}-h (
\hat{x}_{i,k |k-1})\right)^{\T} \Gamma _{i,k}^{-1} \left(y_{k}-h(\hat{x}_{i,k |k-1}) \right)\right], \\
\Gamma_{i,k}&=  \left. \nabla_x h \right|_{\hat{x}_{i, k|k-1} }  P_{i,k|k-1}^x    \left. \nabla_x^\T h \right|_{\hat{x}_{i, k|k-1} }  +R_{k}.   
\end{align}
\end{subequations}
The final, fused estimate  of $\xi_k$, i.e., $\delta $ and
$x_{k}$, and the associated covariance are
\begin{subequations}\label{Fused-Estimation}
\begin{align}
\hat{\xi}_{k|k}&=  \int \xi_k \p(\xi_k | \mathcal{Y}_k ) d  \xi_k  = \overset{M}{\underset{i=1}{\sum }}\mu _{i,k}\hat{\xi}%
_{i,k|k},  \label{mas17} \\
P_{k|k}^\xi &= \int  \left( \xi_k - \hat \xi_{k|k} \right)  \left( \xi_k - \hat \xi_{k|k} \right)^\top \p(\xi_k | \mathcal{Y}_k) d \xi_k  =\overset{M}{\underset{i=1}{\sum }}\mu _{i,k} \left[P_{i,k|k}^\xi +(\hat{\xi}%
_{i,k|k} -\hat{\xi}_{ k|k})(\hat{\xi}_{i,k|k} -\hat{\xi}_{ k|k})^{\T} \right].
\label{mas18}
\end{align}%
\end{subequations}
\end{lemma}

\begin{proof}
See Appendix.\ref{Fusion-Proof}.
\end{proof}

Lemmas \ref{Prediction-Lemma}-\ref{Fusion-Lemma} outline the complete design of a computational SSUE algorithm for the system in \eqref{mas1}, which is further  summarized in Table \ref{Algorithm-Summary}. This algorithm is built within the Bayesian estimation framework laid out in~\eqref{Bayesian-Update} to identify the uncertainty, state, and the uncertainty's location.  

\begin{table*}[tbh]
  \centering
    \caption{The  SSUE algorithm based on the Newton's   method.}\label{Algorithm}
  \framebox[0.6\linewidth]{ 
  \begin{minipage}{0.58\linewidth}
      \begin{algorithmic}{\normalsize
    \State  {\bf Initialize:} $k=0$, $\hat{\delta}_{0|0}$,  $\hat{x}_{0|0}$, $P_{0|0}^\delta$, $P_{0|0}^x$,  $P_{0|0}^{\delta x}$
    \Repeat    
    \State $k \leftarrow k+1$
    \\
    \textit{Prediction:}
    \State State and uncertainty prediction via (\ref{Prediction-Formula})
    \\
    \textit{Update:}
    \State Initialize: $\ell = 0$, $\hat{\xi}_{i,k|k}^{(0)}=%
\begin{bmatrix}
\hat{\delta}_{i,k|k-1}^\top & \hat{x}_{i,k|k-1}^\top%
\end{bmatrix}^\top%
$
    \While{$i< \ell_{\max}$ }
    \State State and uncertainty estimation update via~\eqref{Update-Formula}
    \State $i \leftarrow i+1$
    \EndWhile 
    \\
    \textit{Identification of} $\A$:
    \State Evaluate $\p(\A_i | \Y_k)$ and estimate $\A$ via~\eqref{Location-Evaluation}
    \\
    \textit{Fusion:}
    \State Fuse  state and uncertainty estimation via~\eqref{Fused-Estimation}
    \Until{no more measurements arrive}
}
  \end{algorithmic}
\end{minipage}
}\label{Algorithm-Summary}
\end{table*}

\section{Numerical Simulation} \label{Simulation-Example}

In this section, we present a simulation study  that applies the proposed SSUE algorithm to target tracking, with the model taken from~\cite{Sarkka:CUP:2013}. Consider an object is moving in a two-dimensional plane, with three fixed range sensors deployed to measure the distances between the object and the positions of sensors. The state of the object is
\begin{align*}
x = \begin{bmatrix}
d_x & d_y & \dot d_x & \dot d_y
\end{bmatrix}^\T,
\end{align*}
where $d_x$ and $d_y$ are the object's positions in the $x$-$y$ Cartesian coordinates. We suppose that the discrete-time dynamics of $x$ includes parameter uncertainty, which can be expressed in the form of~\eqref{mas1}. Specifically, $\delta=-0.05$, $A$ and the uncertainty set $\mathbb{A}$ are give by
\begin{align*}
A = \begin{bmatrix}
1 & 0 & Ts & 0 \cr0 & 1 & 0 & Ts \cr 0 & 0 & 1 & 0 \cr 0 & 0 & 0 & 1
\end{bmatrix}, \ 
\A_1 = \begin{bmatrix}
0 & 0 & 0 & 1 \cr 0 & 0 & 0 & 0 \cr 0 & 0 & 0 & 0 \cr 0 & 0 & 0 & 0
\end{bmatrix}, \ 
\A_2 = \begin{bmatrix}
1 & 0 & 0 & 0 \cr 0 & 1 & 0 & 0 \cr 0 & 0 & 0 & 0 \cr 0 & 0 & 0 & 0
\end{bmatrix}, \
\A_3 = \begin{bmatrix}
0 & 0 & 0 & 0 \cr 0 & 0 & 0 & 0 \cr 0 & 0 & 1 & 0 \cr 0 & 0 & 0 & 0
\end{bmatrix}.
\end{align*}
We let the true $\A = \A_2$. The noise covariance is
\begin{align*}
Q = \begin{bmatrix}
{1\over 3}Ts^3 & 0 & {1\over 2}Ts^2 & 0 \cr
0 & {1\over 3}Ts^3 & 0 & {1\over 2}Ts^2 \cr
{1\over 2}Ts^2 & 0 & T_s & 0 \cr
0 & {1\over 2}Ts^2 & 0 & T_s
\end{bmatrix}q,
\end{align*}
where $q=0.05$ is the spectral density of the noise. 
The sensors are located at $(s_x^i, s_y^i)$ for $i=1,2,3$. Sensor $i$ produces   measurement of the range 
$\sqrt{(d_x - s_x^i)^2 + (d_y - s_y^i)^2}
$. The measurement noise has a covariance of $R=2$. 

The proposed SSUE algorithm is applied to perform object tracking in the above setting, with the estimation results summarized in Fig.~\ref{Estimation-Results}. Fig.~\ref{Ex1-mu} shows the evolution of $\mu_{i,k} = \p(\A_i | \Y_k)$ for each $\A_i$, which is meant to determine the uncertainty's location. It is seen that $\mu_2$ associated with the true location $\A_2$ will take the lead after certain time, as a  larger amount of measurement becomes available. The estimation of $\delta$ is shown in Fig.~\ref{Ex1-delta}, which approaches the true value gradually through time.  Further, Figs.~\ref{Ex1-x1}-\ref{Ex1-x4} display the state estimation   variables based on the SSUE algorithm as well as the standard EKF without   uncertainty estimation for the purpose of comparison.  The SSUE achieves an overall satisfactory accuracy   in the estimation of all the state variables. The EKF,  produces comparable estimation of the object's position, because the range measurement data is a transformation of the position variables. However, its estimation of the speed variables is poor in accuracy, as illustrated in Figs.~\ref{Ex1-x3}-\ref{Ex1-x4}.  These results illustrate the effectiveness and utility of the proposed SSUE framework and algorithm.

\begin{figure}[t] \centering
    \subfigure[]{
    	\includegraphics[width=0.35\textwidth]{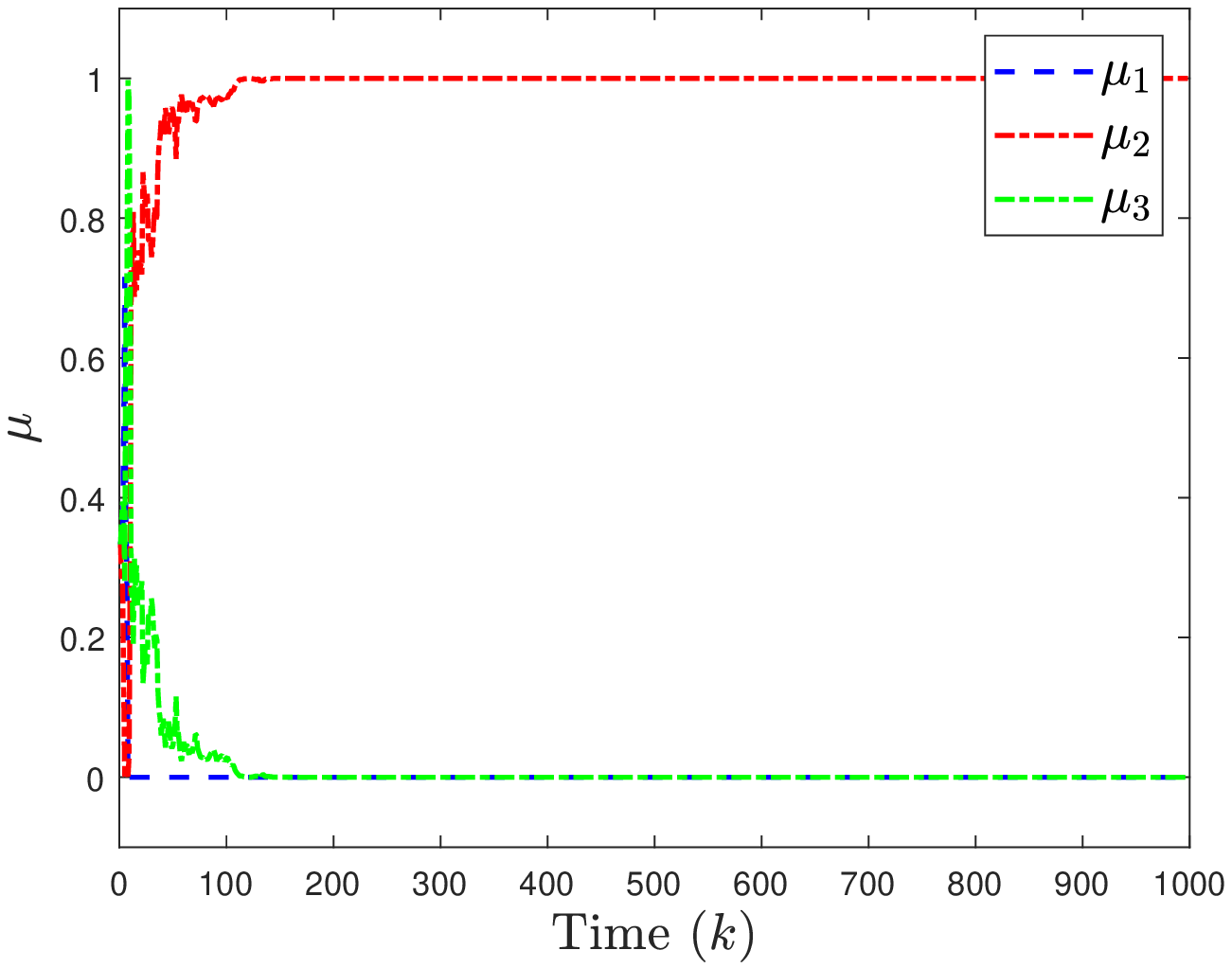}		\label{Ex1-mu}}
    \subfigure[]{
    \includegraphics[width=0.35\textwidth]{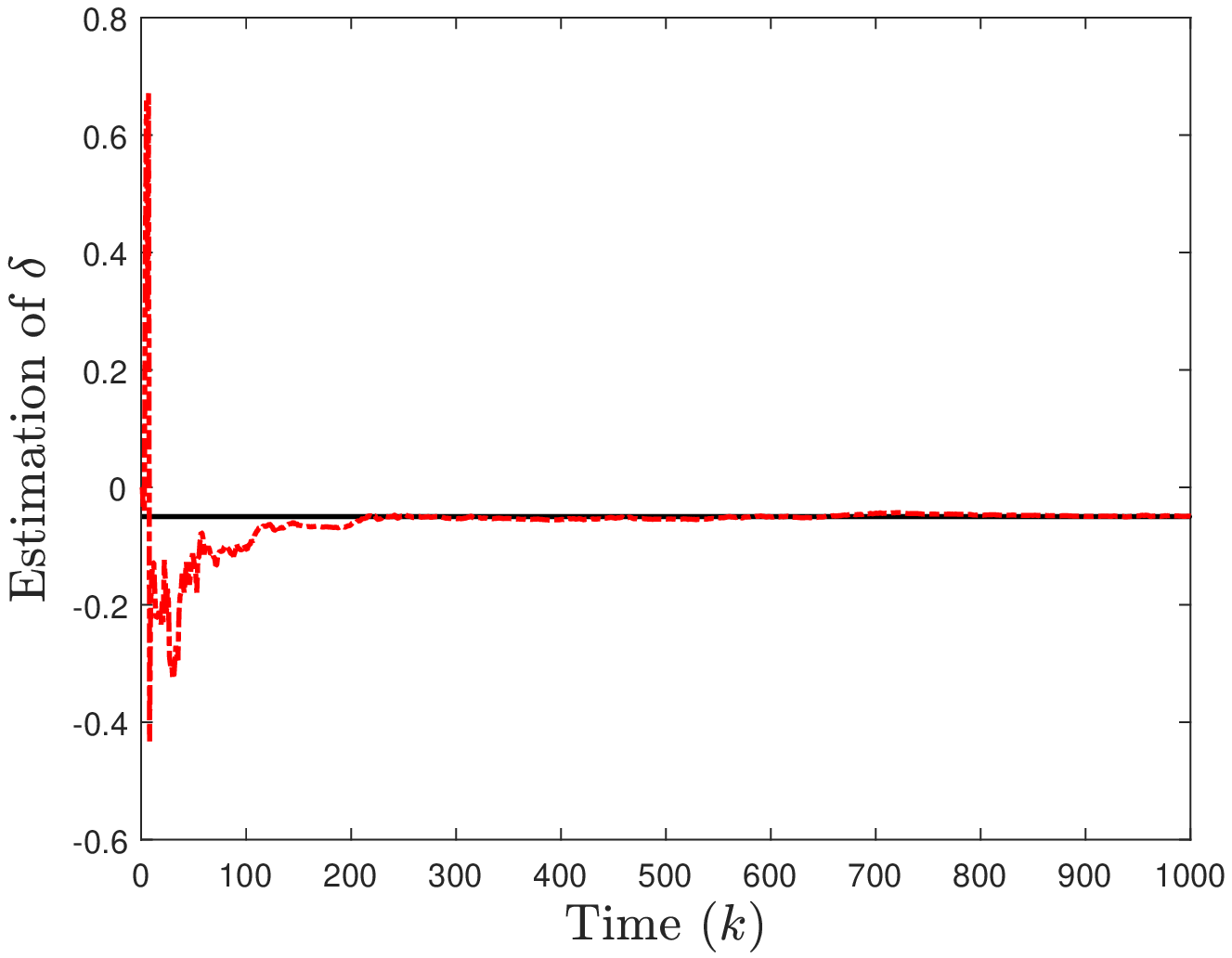}\label{Ex1-delta}}\\
    \subfigure[]{
    \includegraphics[width=0.35\textwidth]{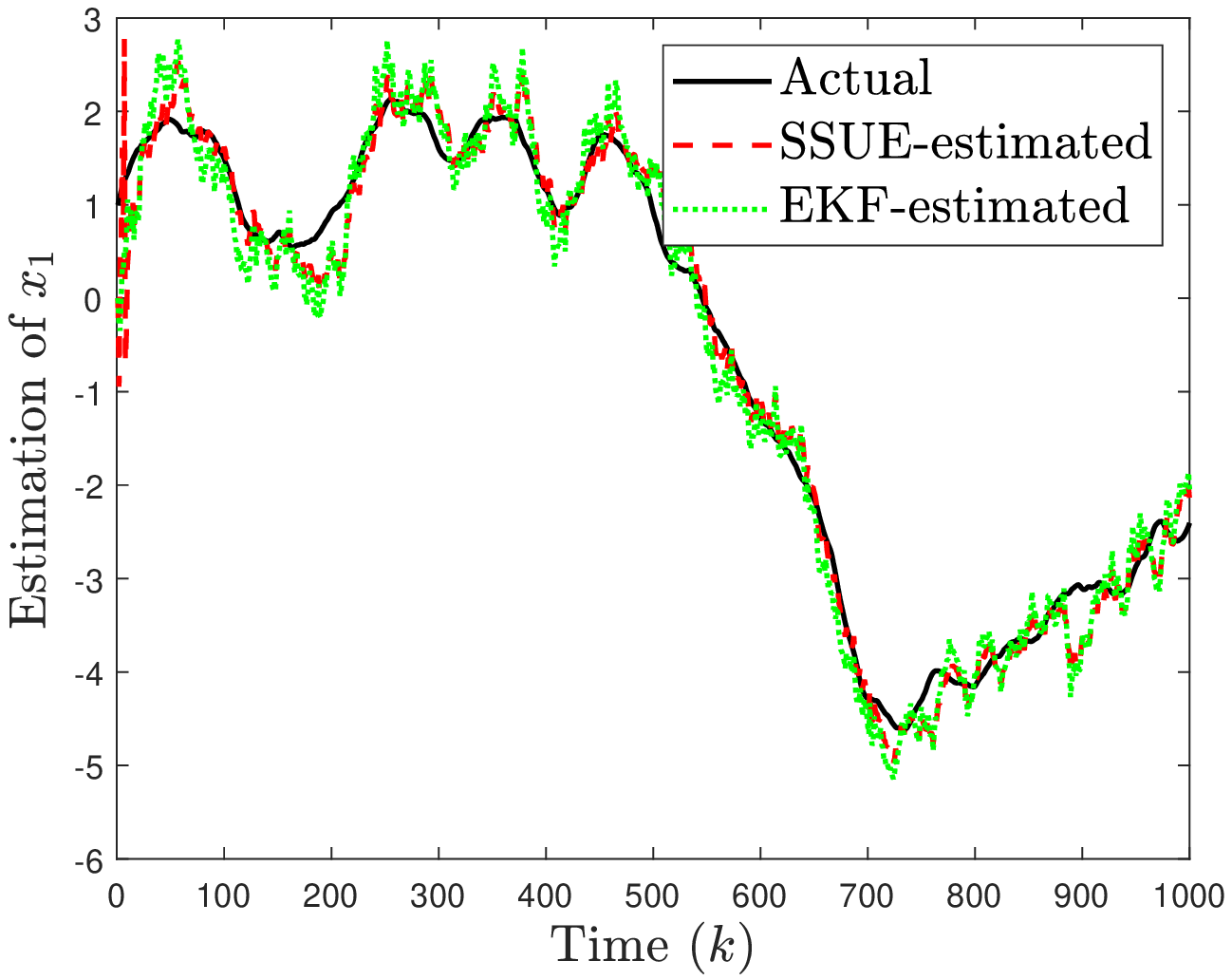}\label{Ex1-x1}}  
    \subfigure[]{
    \includegraphics[width=0.35\textwidth]{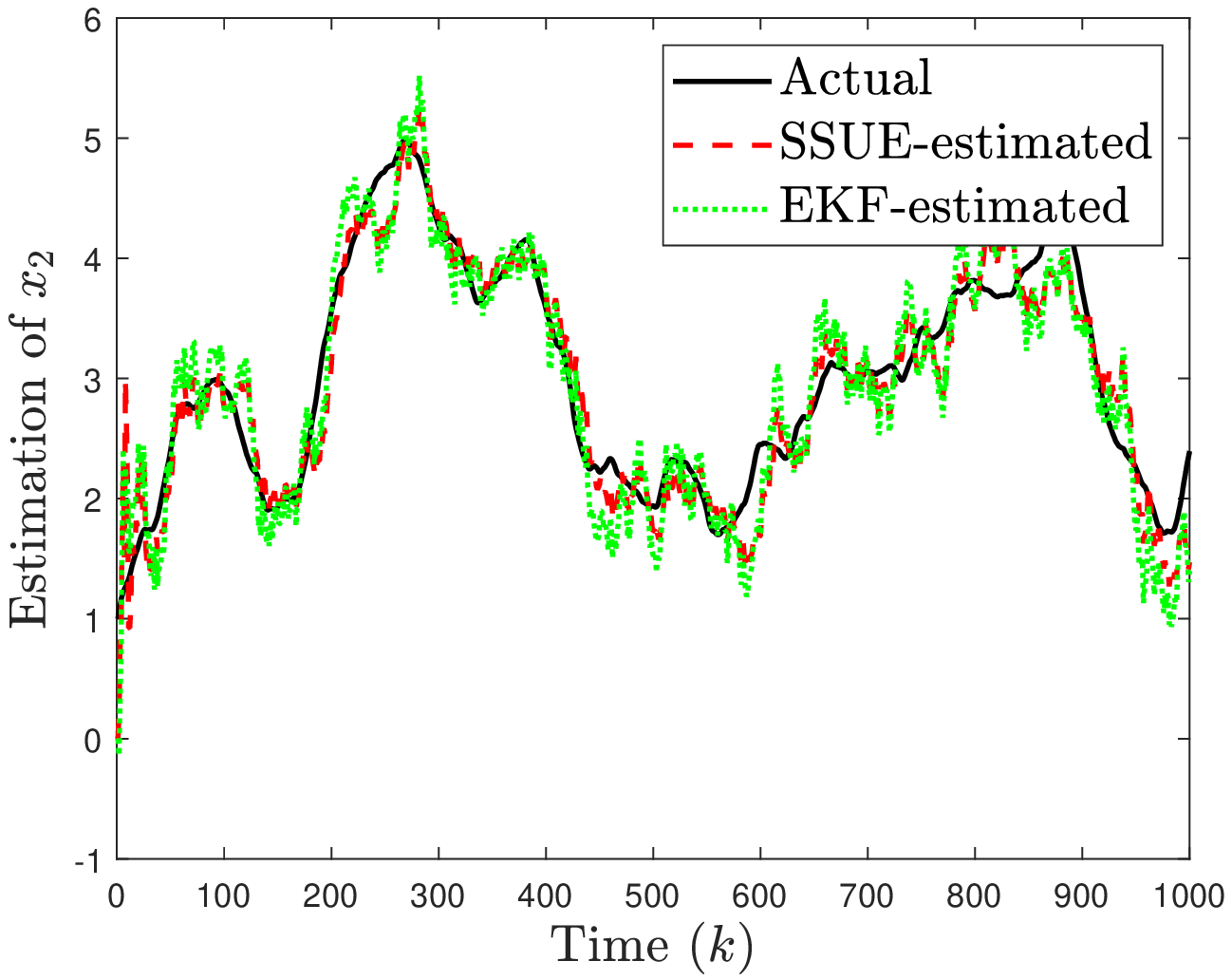}\label{Ex1-x2}} \\
    \subfigure[]{
    \includegraphics[width=0.35\textwidth]{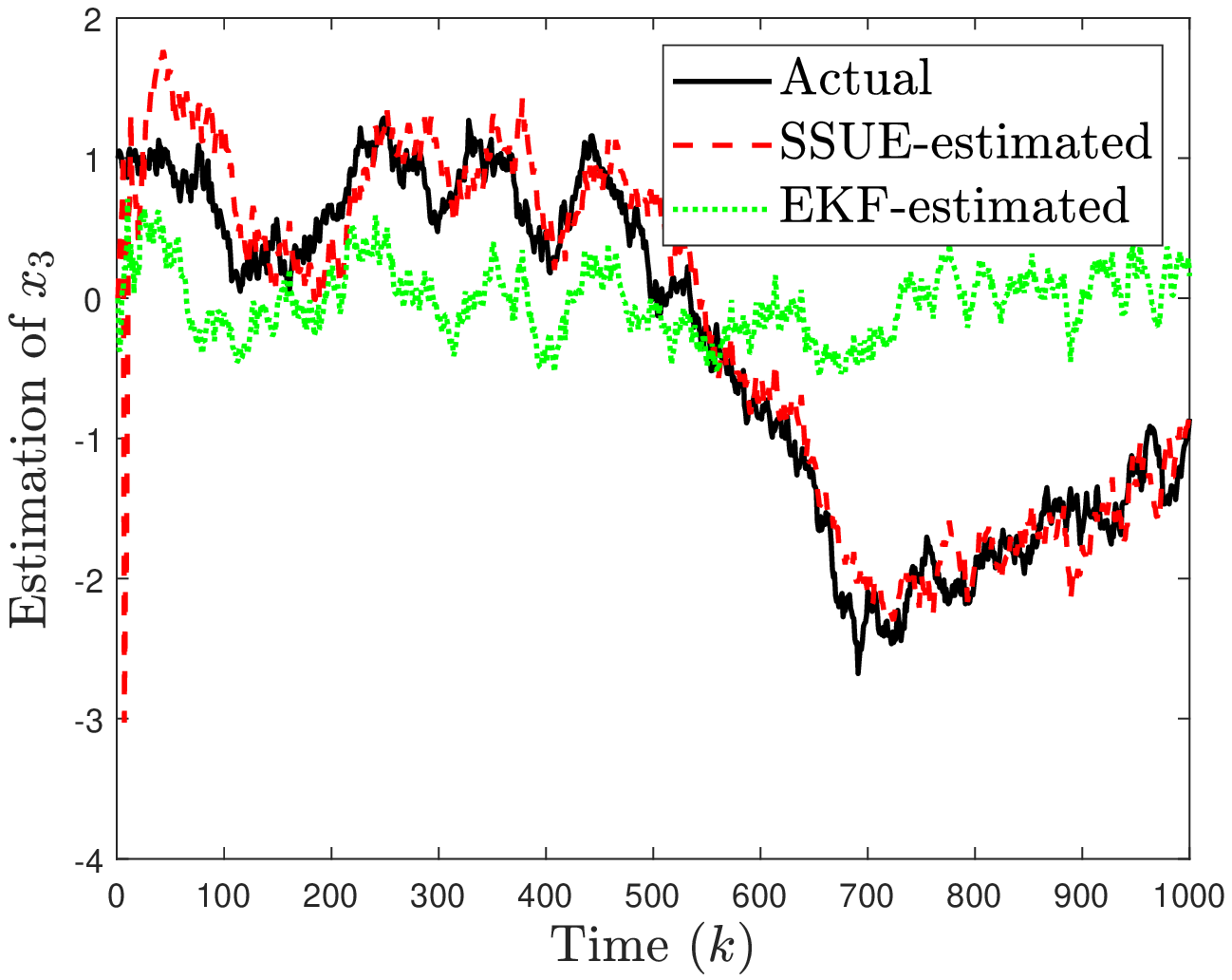}\label{Ex1-x3}} 
    \subfigure[]{
    \includegraphics[width=0.35\textwidth]{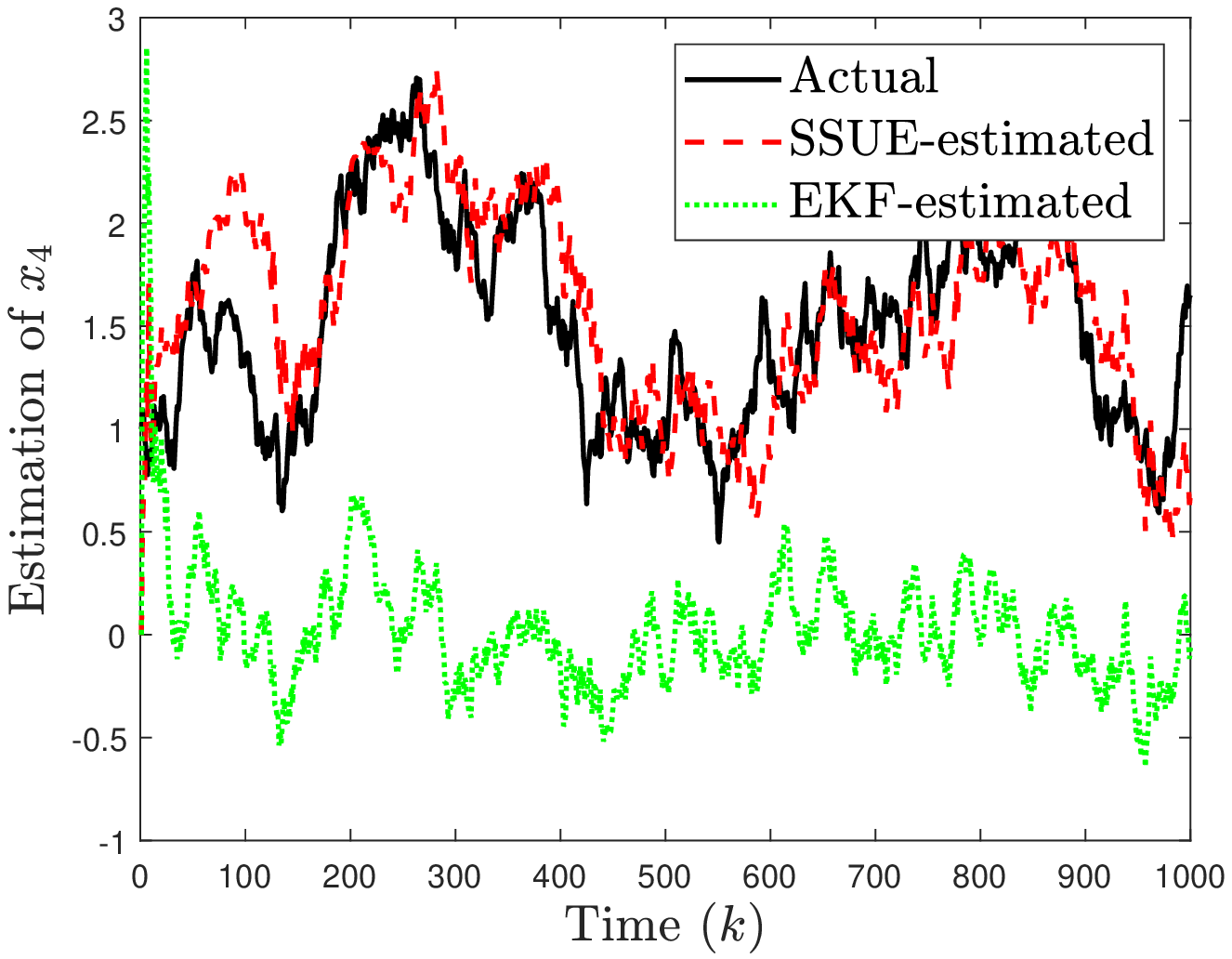}\label{Ex1-x4}} 
  \caption{Simulation with object tracking using range sensors: (a) $\mu_i = \p(\A_i | \Y_k)$; (b) estimation of the  uncertainty $\delta$; (c) estimation of $x_1$; (d) estimation of $x_2$; (e) estimation of $x_3$; (f) estimation of $x_4$.}
\label{Estimation-Results}
\end{figure}

\section{Conclusions}\label{Conclusion}

This paper considered estimation for an uncertain dynamical system and studied the SSUE problem concerned with simultaneously estimating the state, parameter uncertainty, and uncertainty's location hidden in a model.  Based on the thinking of probabilistic estimation, we constructed a Bayesian framework which can update the pdf's of the uncertainty and state conditioned on measurement data and evaluate the probabilities of the possible locations of the uncertainty. Building on framework, we developed  a computational algorithm to perform SSUE as desired, which was based on MAP estimation and the Newton's method for numerical optimization. We investigated the observability for SSUE  and further revealed the relation between observability and consistency of the estimation of the uncertainty's location. A simulation example demonstrated the effectiveness of the proposed SSUE design.

\section*{Appendix}

\renewcommand{\thesubsection}{\Alph{subsection}}

\setcounter{equation}{0}

\renewcommand{\theequation}{IA.\arabic{equation}}

\subsection{Preliminary}

\begin{lemma}\label{AAT-BBT}
\cite{Horn:CUP:2013}For $A,B \in \mathbb{R}^{m\times n}$,  $AA^\T = BB^\T$ if and
 only if there is a unitary $U\in \mathbb{R}^{n \times n}$ such that $A=BU$.
\end{lemma}

\subsection{Proof of Theorem~1} \label{Bayesian-Paradigm-Proof}

\begin{proof}
By the Chapman-Kolmogorov equation and Bayes' rule, we have
\begin{align*}
\p(\delta ,x_{k}|\mathscr{A}_{i},\mathcal{Y}_{k-1})=
\int \p(\delta, x_{k}
,x_{k-1}|\A_i,\mathcal{Y}_{k-1} ) dx_{k-1} ,
\end{align*}
which, according to   Bayes' rule, can reduce to~\eqref{mas6}, i.e.,
\begin{align*}
\p(\delta ,x_{k}|\mathscr{A}_{i},\mathcal{Y}_{k-1})
&= \int \frac{p(\delta ,x_{k}, x_{k-1},\mathscr{A}_{i},\mathcal{Y}_{k-1})}{p(\mathscr{A}_{i},\mathcal{Y}_{k-1})} dx_{k-1}\\
&= \int \frac{\p(x_{k} | \delta, x_{k-1},\mathscr{A}_{i},\mathcal{Y}_{k-1})   p( \delta, x_{k-1},\mathscr{A}_{i},\mathcal{Y}_{k-1})}{p(\mathscr{A}_{i},\mathcal{Y}_{k-1})} dx_{k-1}\\
&=\int \p( x_{k}|\delta ,x_{k-1}, \A_i  )   \p( \delta ,x_{k-1}|\mathscr{A}_{i},
\mathcal{Y}_{k-1} ) d x_{k-1}.
\end{align*}
The derivation uses $\p( x_{k}|\delta
,x_{k-1}, \A_i )  = \p( x_{k}|\delta
,x_{k-1},\A_i,\mathcal{Y}_{k-1} ) $, which results from the Markovian dependence of $x_{k}$  on $\delta $ and $x_{k-1}$ given $\A_i$.  When  $y_{k}$ is available, it can be used to update $\p(\delta ,x_{k}|\mathscr{A}_{i},\mathcal{Y}_{k-1})$ to $\p(\delta ,x_{k}|\mathscr{A}_{i},\mathcal{Y}_{k})$.  Using   Bayes' rule again, we can obtain
\begin{align*}
\p ( \delta ,x_{k}|\mathscr{A}_{i},\mathcal{Y}_{k} ) & =\frac{%
\p ( y_{k}|\delta ,x_{k},\mathscr{A}_{i},\mathcal{Y}_{k-1} ) \p(
\delta ,x_{k}|\mathscr{A}_{i},\mathcal{Y}_{k-1} ) }{\p ( y_{k}|%
\mathcal{Y}_{k-1} ) }~~  \notag \\
& =\frac{\p ( y_{k}|x_{k},\mathscr{A}_{i} ) \p ( \delta ,x_{k}|%
\mathscr{A}_{i},\mathcal{Y}_{k-1} ) }{\p ( y_{k}|\mathcal{Y}%
_{k-1} ) }  \notag \\
& \varpropto \p ( y_{k}|x_{k}  ) \p ( \delta
,x_{k}|\mathscr{A}_{i},\mathcal{Y}_{k-1} ),
\end{align*}%
which verifies~\eqref{mas7}. In above, the second equality is due to the  dependence of $y_{k}$ on $%
x_{k}$. Note that $\p ( y_{k}|\mathcal{Y}_{k-1} )
$ only plays a role of proportional coefficient. Finally,
\eqref{mas5}-\eqref{mas4} follow directly from the law of total probability and Bayes' rule. 
\end{proof}

\subsection{Proof of Theorem~3}\label{Consistency-Proof}
\begin{proof}
Consider the system in~\eqref{mas1}. Because $x_0 \sim \mathcal{N} (0, P_0)$, $w_k \sim \mathcal{N} (0, Q_k)$, and $v_k \sim \mathcal{N} (0, R_k)$  are independent of each other,  we can obtain 
\begin{equation*}
\mathcal{Y}_{k}=\mathcal{O}_{k}(\delta ,\mathscr{A})x_{0}+\mathcal{I}%
_{k} (\delta ,\mathscr{A}  )W_{k}+V_{k},
\end{equation*}%
where%
\begin{align*}
\mathcal{I}_{k}(\delta,\mathscr{A})=&%
\begin{bmatrix}
0 & 0 & 0 & \cdots & 0 \\
C & 0 & 0 & \cdots & 0 \\
C(A+\delta \mathscr{A}) & C & 0 & \cdots & 0 \\
\vdots & \vdots & \vdots & \ddots & \vdots \\
C(A+\delta \mathscr{A})^{k-1} & C(A+\delta \mathscr{A})^{k-2} & C(A+\delta \mathscr{A})^{k-3}  & \cdots &
C%
\end{bmatrix},
\\
W_{k}=&%
\begin{bmatrix}
w _{0}^{\T} & w _{1}^{\T} & \cdots & w _{k-1}^{\T}
\end{bmatrix}
^{\T},V_{k}=%
\begin{bmatrix}
v _{0}^{\T} & v _{1}^{\T} & \cdots & v _{k}^{\T}%
\end{bmatrix}%
^{\T}.
\end{align*}
It is seen that $\Y_k$ follows a zero-mean Gaussian distribution.
Let $\Sigma_k$  denote its covariance, i.e., $\Sigma _{k} = \mathbb{E}\left\{ \mathcal{Y}_{k}\mathcal{Y}_{k}^{\T} \right\}$. Then, 
\begin{align*}
\Sigma _{k}(\delta ,\mathscr{A})&= \Pi_k(\delta ,\mathscr{A}) \Omega_k \Pi_k(\delta ,\mathscr{A})^\T +\mathcal{R}_k,
\end{align*}%
where
\begin{align*}
\Pi_k(\delta ,\mathscr{A}) = 
\begin{bmatrix} 
\mathcal{O}_{k}(\delta ,\mathscr{A}) &  \mathcal{I}_{k}(\delta ,\mathscr{A}) 
\end{bmatrix}, \ 
\Omega_k = \mathrm{diag} \left( P_0, Q_1, Q_2,\ldots,Q_k \right),
\end{align*}
where $\mathcal{Q}_{k}=\mathrm{diag}(Q_0,Q_1,\ldots,Q_{k-1})$,$\mathcal{R}%
_{k}= \mathrm{diag}(R_{0},R_{1},\ldots,R_{k})$. Obviously, $\Omega_k \geq 0$ and  is a real symmetric matrix since $P_0>0$ and $Q_k \geq 0$. It thus admits eigen-decomposition, $\Omega_k = U\Lambda U^\T$, where $U$ is an orthogonal matrix whose columns are the eigenvectors of $\Omega_k$ and $\Lambda$  is a diagonal matrix whose entries are the eigenvalues of $\Omega_k$. Note that the first $n$ diagonal entries of $\Lambda$, corresponding to the eigenvalues of $P_0$, must be positive. 
We then have $\Sigma _{k}(\delta ,\A ) = \Phi \Phi^\top + \mathcal{R}_k$, where $\Phi =  \Pi _{k}(\delta ,\A )U \Lambda^{1 \over 2} $. 
Now, consider $\A_i,\A_j \in \mathbb{A}$ with $i\neq j$, and without loss of generality, consider the same $\delta$ for both $\A_i$ and $\A_j$. By  Theorem~\ref{Observability-testing},  there exists $K$ such that, for $k \geq K$,
\begin{align*}
\text{rank}\left(
\begin{bmatrix}
\mathcal{O}_{k}(\delta,\mathscr{A}_i) & \mathcal{O}_{k}(\delta,%
\mathscr{A}_j)%
\end{bmatrix}%
\right)=2n  \ \Rightarrow \ \text{rank}\left(
\begin{bmatrix}
\Pi_{k}(\delta,\mathscr{A}_i) & \Pi_{k}(\delta,
\mathscr{A}_j)%
\end{bmatrix}%
\right) \geq 2n.
\end{align*}
Given this and using Lemma~\ref{AAT-BBT}, we obtain 
\begin{align}\label{Sigma-Compare}
\Sigma _{k}(\delta , \A_i) \neq \Sigma _{k}(\delta , \A_j),
\end{align}
which suggests that the measurement sequences corresponding to $\A_i$ and $\A_j$ have different covariances.  
For notational convenience, we write  $\Sigma _{k}(\delta , \A_i)$ as $\Sigma _{k}^i$. We hence have $\Sigma_k^t \neq \Sigma_k^i$ for $t \neq i$, and then from the properties of the KL divergence~\cite{Kullback}, it   follows  that    $\D(\mathscr{A}_{t} \| \mathscr{A}_{i})> 0$. This establishes~\eqref{KL-Divergence}.

Because
\begin{align*}
h_{i}^{t} (\mathcal{Y}_{k}  )=\frac{\p ( \mathcal{Y}_{k}|\mathscr{A}_{t} ) }{\p (
\mathcal{Y}_{k}|\mathscr{A}_{i} ) } 
=\prod_{\tau =1}^{k}\frac{\p( y_{\tau }|%
\mathscr{A}_{t},\mathcal{Y}_{\tau -1}
) }{\p ( y_{\tau }|\mathscr{A}%
_{i},\mathcal{Y}_{\tau -1} ) }=\prod_{\tau =1}^{k}h_{i}^{t} (y_{\tau } \, | \,
\mathcal{Y}_{\tau -1} ),
\end{align*}
we have
\begin{align} \label{ratio-t-i-Y_k}
 \log h_{i}^{t}  (%
\mathcal{Y}_{k}  ) = \sum_{\tau=1}^k \log h_{i}^{t} (y_{\tau } \, | \, \mathcal{%
Y}_{\tau -1}  ).
\end{align}
Assumption \ref{assum2} stipulates that the sequence $\left\{ \log h_{i}^{t}  (y_{\tau } \, | \, \mathcal{%
Y}_{\tau -1}  ) \right\}$ is ergodic. This implies
\begin{align*}
\lim_{k \rightarrow \infty }\frac{1}{k}%
\sum_{\tau =1}^{k}\log h_{i}^{t}  (y_{\tau} \, | \, \mathcal{Y}_{\tau-1}  )
= \mathbb{E}_{t} \left \{ \log h_{i}^{t}  (y_{\tau } \, | \, \mathcal{Y}_{\tau
-1}  ) \right\}  \geq 0,
\end{align*}
from which one can find out that
\begin{align}\label{ratio-t-i-Y_k-Y_k-1}
\lim_{k\rightarrow \infty}\sum_{\tau =1}^{k}\log h_{i}^{t} (y_{\tau} \, | \, \mathcal{Y}_{\tau-1}  )>0.
\end{align}
Note that the case does not hold that $\lim_{k\rightarrow \infty}\sum_{\tau =1}^{k}\log h_{i}^{t}  (y_{\tau} \, | \, \mathcal{Y}_{\tau-1}  )=0$, because   
\begin{align*}
\D(\mathscr{A}_{t} \| \mathscr{A}_{i})=& \mathbb{E}_{t} \left(  \sum_{\tau =1}^{k} \log h_{i}^{t} (y_{\tau } \, | \, \mathcal{%
Y}_{\tau -1} )\right)>0 .
\end{align*}
By~\eqref{ratio-t-i-Y_k}-\eqref{ratio-t-i-Y_k-Y_k-1}, we obtain
\begin{align*}
\lim_{k\rightarrow \infty} \log h_{i}^{t} (\mathcal{Y}_{k}  ) >0.
\end{align*}
This indicates  that $\lim_{k\rightarrow \infty}   h_{i}^{t} \left(\mathcal{Y}_{k} \right)>1 $, verifying~\eqref{Ratio-Larger-than-1}.

Let us continue to consider the  sequence $\left\{ h_{t}^{i} (
\mathcal{Y}_{k} )\right\} $. 
It is a positive discrete-time martingale with respect to $\Y_k$~\cite{Doob}. Then, by  the martingale convergence theorem, there exists a random variable $ 0 \leq H_i <\infty$ such that $h_{t}^{i} (\mathcal{Y}_{k}  )$ converges to $H_i$ a.s. as $k \rightarrow \infty$, i.e.,
\begin{align*}
\lim_{k \rightarrow \infty} h_{t}^{i}  (
\mathcal{Y}_{k}  ) = H_i, \ \mathrm{a.s.}
\end{align*}
Let us now prove that $H_i=0$ by contradiction. Assume that $H_i\neq 0$. This implies that the probability of the following event is nonzero: 
\begin{align*}
\lim_{k \rightarrow \infty} h_{t}^{i}  ( y_k \, | 
 \, \mathcal{Y}_{k-1}  ) = \lim_{k \rightarrow \infty} \frac{ h_{t}^{i}  ( \Y_k  )}{ h_{t}^{i}  ( 
\mathcal{Y}_{k-1}  ) } = 1. 
\end{align*}
This then gives rise to
\begin{align*}
\lim_{k \rightarrow \infty} \frac{\p(y_{k+K},  \ldots, y_{k+1}| \Y_k, A_t )}{\p(y_{k+K}, \ldots, y_{k+1}| \Y_k, A_i) } =  \lim_{k \rightarrow \infty} \prod_{\tau=k+1}^{k+K } h_i^t  (y_{\tau} | \Y_{\tau-1}  ) = 1,
\end{align*}
which shows that the KL divergence between $\p(y_{k+K},  \ldots, y_{k+1}| \Y_k, A_t )$ and $\p(y_{k+K}, \ldots, y_{k+1}| \Y_k, A_i)$ will approach zero as $k \rightarrow \infty$. This, however, can be proven untrue along  similar lines to  the proof of~\eqref{KL-Divergence}. Therefore, $H_i=0$ and $\lim_{k \rightarrow \infty} h_{t}^{i}  (\mathcal{Y}_{k}  )$, a.s. This then yields~\eqref{Ratio-infty}.  
\end{proof}

\subsection{Proof of Lemma~1}\label{Prediction-Proof}

\begin{proof} Applying the first-order Taylor series expansion to the right-hand side of  the process equation,  we have
\begin{align*}
x_k \approx (A+\hat \delta_{k-1|k-1} \A_i) \hat x_{i, k-1|k-1} + F_{k-1} \begin{bmatrix}
\delta - \hat \delta_{k-1|k-1} \cr x_k - \hat x_{k-1|k-1}
\end{bmatrix} + w_k.
\end{align*}
Now $x_k$ is approximated as a linear transformation of $\delta$ and $x_k$. Then,  $\p(x_k | \A_i, \Y_{k-1})$   can be approximated by a Gaussian one based on~\eqref{P_delta_x_Y_k_Gaussian}. We can  then use~\eqref{mas6} to determine that 
\begin{align*}
\p(\delta, x_k | \A_i, \Y_{k-1}) = \mathcal{N} \left( 
\begin{bmatrix}
\hat \delta_{i,k|k-1} \cr \hat x_{i, k|k-1} 
\end{bmatrix},
\begin{bmatrix}
P_{i,k|k-1}^{\delta } & P_{i,k|k-1}^{\delta x } \cr 
\left( P_{i,k|k-1}^{\delta  x}\right)^\T & P_{i,k|k-1}^{x }
\end{bmatrix}
\right).
\end{align*}
The solution to~\eqref{Prediction-MAP} thus is given by the mean of $\p(\delta, x_k | \A_i, \Y_{k-1})$, leading to the results in this lemma.
\end{proof}

\subsection{Proof of Lemma~2} \label{Update-Proof}

\begin{proof}
Based on~\eqref{mas7} and  Assumption~\ref{assum4},  we can derive the log-MAP cost function for $p( \A_i | \Y_{k})$  as 
\begin{equation*}
L_i(\xi_k )=\alpha
_{k}^{\T}R_{k}^{-1}\alpha _{k}+\beta_{i,k}^{\T} \left(P_{k|k-1}^\xi\right)^{-1}\beta_{i,k},
\end{equation*}%
where $
\alpha _{k}=y_{k}-h (x_{k})$ and
$\beta _{i,k}=
\xi_{k}-\hat{\xi}_{i,k|k-1}$.
The problem in~\eqref{Update-Estimation} then becomes
\begin{align}\label{Update-Min-Problem}
\hat \xi_{i,k|k} = \arg \min_{\xi_k} L(\xi_k).
\end{align}
Note that one can write $L_i(\xi_k)$  as $L_i(\xi_k) = r_i^\T (\xi_{i,k}) \cdot r_i (\xi_{i,k})$, implying that~\eqref{Update-Min-Problem} is a nonlinear least squares problem, which can be addressed using the  Newton's method as shown in~\eqref{Newton-Iterative-Update}.

The estimation error covariance of an MAP estimator is the inverse of the Fisher information matrix. Hence,
\begin{align*}
P_{i, k|k}^\xi = \mathcal{F}^{-1} (\xi_{i, k|k}),
\end{align*}
where $\mathcal{F}$ is the Fisher information matrix defined as
\begin{align*}
\mathcal{F} = \mathbb{E} \left\{ \nabla_\xi L_i  \cdot \nabla_\xi ^\T L_i \right\}.
\end{align*}
Because $ \nabla_\xi L_i =  \nabla_\xi^\T r_i  \cdot r_i$, we have
\begin{align*}
\mathcal{F} = \nabla _{\xi
}^{\T}r_{i} \cdot \nabla _{\xi }r_{i},
\end{align*}
which leads to~\eqref{mas16}.
\end{proof}

\subsection{Proof of Lemma~3}\label{Fusion-Proof}

\begin{proof} 
Define $\mu_{i,k} =  \p(\A_i | \Y_k)$ and $\lambda_{i,k}  = \p(y_k | \A_i , \Y_{k-1})$. This implies~\eqref{Location-Identification} from~\eqref{Location-Identification-Initial}.  
Further, according to Assumption 4, we have $\p(y_{k}|x_{k})  \sim \mathcal{N}\left( h(x_{k}),R_{k}\right) $. It can be approximated as 
\begin{align}\label{y_x_prob}
\p(y_k | x_k  )  \sim \mathcal{N} \left(h(\hat x_{i, k|k-1}) +  \left.   \nabla_x h \right|_{\hat{x}_{i, k|k-1}} \left( x_k - \hat x_{i,k|k-1} \right),    R_k  \right),
\end{align}
 if $h(x_k)$ is Taylor-expanded as
$
h(x_k) \approx h (\hat x_{i, k|k-1}) +   \left.  \nabla_x h \right|_{\hat{x}_{i, k|k-1}} \left( x_k - \hat x_{i,k|k-1} \right)
$. 
Meanwhile, Assumption 4 also indicates 
\begin{align}\label{x_Y_prob}
\p (x_k | \A_i,  \Y_{k-1})  \sim \mathcal{N}\left( \hat x_{i,k|k-1}, P_{i,k|k-1}^x \right) .
\end{align}
Based on~\eqref{y_x_prob}-\eqref{x_Y_prob}, we can readily obtain  
\begin{align*}
\lambda_{i,k} &=   \p(y_k | \A_i , \Y_{k-1}) = \int  \p(y_k | x_k ) \p(x_k |   \A_i , \Y_{k-1}) d x_k  \sim \mathcal{N}\left(   h(\hat x_{i, k|k-1}),  \Gamma_{i, k } \right),
\end{align*}
which gives~\eqref{lambda-update}. By~\eqref{mas5}, the update of $\mu_{i,k}$   then is governed by~\eqref{mu-update}. Finally,~\eqref{Fused-Estimation}   follows from~\eqref{mas4}. 
\end{proof}

  \bibliographystyle{WileyNJD-Harvard}

\bibliography{References}

\end{document}